\pdfoutput=1

\documentclass[11pt,letterpaper]{article}
\usepackage[hscale=0.7,vscale=0.8]{geometry}
\usepackage{mystyle} 
\pgfplotsset{compat=1.14}

\title{Black-Box Strategies and Equilibrium for Games with Cumulative Prospect Theoretic Players}
\author{Soham R.\ Phade and Venkat Anantharam
\thanks{Research supported by  
the NSF grants CNS--1527846, CCF--1618145 and CCF-1901004,
by the NSF Science \& Technology
Center for Science of Information Grant number CCF-0939370
and by the William and Flora Hewlett Foundation 
supported Center for Long Term Cybersecurity at Berkeley.
}
\thanks{The authors are with the Department of Electrical Engineering and Computer Sciences, University of California, Berkeley, Berkeley, CA 94720.
        {\tt\small soham\_phade@berkeley.edu, ananth@eecs.berkeley.edu}}%
}

\newcommand{\BAR}{\cal{A}}
\newcommand{\BBR}{\cal{B}}
\newcommand{\co}{\overline{co}}

\newcommand{\pNE}{\mathrm{pNE}}
\newcommand{\mNE}{\mathrm{mNE}}
\newcommand{\BBNE}{\mathrm{BBNE}}
\newcommand{\mBBNE}{\mathrm{mBBNE}}

\newcommand{\cob}{\color{black}}

\begin{document}
\maketitle


\abstract{
	The betweenness property of preference relations states that a probability mixture of two lotteries should lie between them in preference. 
	It is a weakened form of the independence property and hence satisfied 
	in expected utility theory (EUT). 
	Experimental violations of betweenness are well-documented and several preference theories, notably cumulative prospect theory (CPT), do not satisfy betweenness. 
	We prove that CPT preferences satisfy betweenness if and only if they conform with EUT preferences.
	In game theory, lack of betweenness in the players' preference relations makes it essential to distinguish between the two interpretations of a mixed action by a player -- conscious randomizations by the player and the uncertainty in the beliefs of the opponents. 
	We elaborate on this distinction and study its implication 
	for the definition of Nash equilibrium. 
	This results in four different notions of equilibrium, with pure and mixed action Nash equilibrium being two of them. 
	We dub the other two pure and mixed black-box strategy Nash equilibrium respectively.
	We resolve the issue of existence 
	of such equilibria and examine how these different notions of equilibrium compare with each other.
}


\section{Introduction}
\label{sec: intro}

There is a large amount of evidence that human agents as decision-makers do not conform to the independence axiom of expected utility theory (EUT).
(See, for example, \citet{allais1953extension,weber1987recent} and \citet{machina1992choice}.)
This has led to the study of several 
alternate theories 
that do away with the independence axiom \citep{machina2014nonexpected}.
Amongst these, the cumulative prospect theory (CPT) of \citet{tversky1992advances} stands out since it accommodates many of 
the empirically observed behavioral features 
from human experiments
without losing much analytical tractability \citep{wakker2010prospect}.
Further, it includes EUT as a special case.

The independence axiom says that if lottery $L_1$ is weakly preferred over lottery $L_2$ by an agent (i.e. the agent wants lottery $L_1$ at least as much as lottery $L_2$), and $L$ is some other lottery, then, for $0 \leq \alpha \leq 1$, the combined lottery $\alpha L_1 + (1-\alpha) L$ is weakly preferred over the combined lottery $\alpha L_2 + (1- \alpha) L$ by that agent.
A weakened form of the independence axiom, called betweenness,
says that if lottery $L_1$ is weakly preferred over lottery $L_2$ (by an agent), then, for any $0 \leq \alpha \leq 1$, the mixed lottery $L = \alpha L_1 + (1 - \alpha) L_2$ must lie between the lotteries $L_1$ and $L_2$ in preference.
Betweenness implies that if an agent is indifferent between $L_1$ and $L_2$, then she is indifferent between any mixtures of them too.
It is known that independence implies betweenness, but betweenness does not imply independence \citep{chew1989axiomatic}.
As a result, EUT preferences, which are known to satisfy the independence axiom, also satisfy betweenness. 
CPT preferences, on the other hand, do not satisfy betweenness in general (see example~\ref{ex: non-monotone_lotteries}).
In fact, in theorem~\ref{thm: CPTbet=indep}, we show that CPT preferences satisfy betweenness if and only if they are EUT preferences
(recall that EUT preferences are a special case of CPT preferences).
Several empirical studies show systematic violations of betweenness \citep{camerer1994violations,agranov2017stochastic,dwenger2012flipping,sopher2000stochastic},
and this makes the use of CPT more attractive than EUT for modeling human preferences.
Further evidence comes from \citet{camerer1994violations}, where the authors fit data from nine studies using three non-EUT models, one of them being CPT, to find that, compared to the EUT model, the non-EUT models perform better.


Suppose in a non-cooperative game that 
given her beliefs about the other players, a player is indifferent between two of her actions.
Then according to EUT, she should be indifferent between any of the mixtures of these two actions.
This facilitates the proof of the existence of a Nash equilibrium in mixed actions for 
such games.
However, with CPT preferences, the player could either prefer some mixture of these two actions over the individual actions or vice versa.

As a result, it is important to make a distinction
in CPT 
regarding
whether the players can actively randomize over their actions or not. 
One way to 
enable active randomization
is by assuming that 
each player has access
to a randomizing device and the 
player
can ``commit'' to the outcome of this randomization. 
The commitment assumption is necessary, as is evident from the following scenario (the gambles presented below appear in \citet{prelec1990pseudo}). 
Alice needs to choose between the following two actions:
\begin{enumerate}
	\item Action $1$ results in a lottery $L_1 = \{( 0.34,\$20{,}000); (0.66,\$0)\}$, i.e. she receives $\$20{,}000$ with probability $0.34$ and nothing with probability $0.66$.
	\item Action $2$ results in a lottery $L_2 = \{( 0.17,\$30{,}000); (0.83,\$0)\}$.
\end{enumerate}
(See example~\ref{ex: Alice_game} for an instance of a $2$-player game with Alice and Bob, where Alice has two actions that result in the above two lotteries.) 
Note that $L_1$ is a less risky gamble with a lower reward and $L_2$ is a more risky gamble with a higher reward. 
Now consider a compound lottery $L = 16/17 L_1 + 1/17 L_2$.
Substituting for the lotteries $L_1$ and $L_2$ we get $L$ in its reduced form to be 
$$L = \{(0.01,\$30{,}000); (0.32,\$20{,}000); (0.67,\$0)\}.$$
In example~\ref{ex: Alice_nonbetweenness}, we provide a CPT model for Alice's preferences that result in lottery $L_1$ being preferred over lottery $L_2$, whereas lottery $L$ is preferred over lotteries $L_1$ and $L_2$.
Roughly speaking, the underlying intuition is that
Alice is risk-averse in general, and she prefers lottery $L_1$ over lottery $L_2$.
However, she overweights the small $1\%$ chance of getting $\$30{,}000$ in $L$
and finds it lucrative enough to make her prefer lottery $L$ over both the lotteries $L_1$ and $L_2$. 
Let us say Alice has a biased coin that she can use to implement the randomized strategy. 
Now, if Alice tossed the coin, and the outcome was to play action 2, then in the absence of commitment, she will switch to action 1, since she prefers lottery $L_1$ over lottery $L_2$. 
Commitment can be achieved, for example, by asking a trusted party to implement the randomized strategy for her or use a device that would carry out the randomization and implement the outcome without further consultation with Alice. 
Regardless of the implementation mechanism, we will call such randomized strategies 
\emph{black-box strategies}.
The above problem of commitment is closely related to the problem of using non-EUT models in dynamic decisions.
For an interesting discussion on this topic, see \citet[Appendix~C]{wakker2010prospect} and the references therein.

Traditionally, mixed actions have been considered from two viewpoints, especially in the context of mixed action Nash equilibrium.
According to the first viewpoint, these are conscious randomizations by the players -- each player only knows her mixed action and not its pure realization.
The notion of black-box strategies captures this interpretation of mixed actions.
According to the other viewpoint, players do not randomize, and each player chooses some definite action,
but the other players need not know which one, and the mixture 
represents 
their uncertainty, i.e. their conjecture 
about her choice.
\cite{aumann1995epistemic} establish mixed action Nash equilibrium as an equilibrium in conjectures provided they satisfy certain epistemic conditions regarding the common knowledge amongst the players.

In the absence of the betweenness condition, these two viewpoints give rise to different notions of Nash equilibria.
Throughout we assume that the player set and their corresponding action sets and payoff functions, as well as the rationality of each player, are common knowledge.
A player is said to be rational if, given her beliefs and her preferences, she does not play any suboptimal strategy.
Suppose each player plays a fixed action, and these fixed actions are common knowledge, then we get back the notion of \emph{pure Nash equilibrium} (see definition~\ref{def: pureNE}).
If each player plays a fixed action, but the other players have mixed conjectures over her action, and these conjectures are common knowledge, then this gives us \emph{mixed action Nash equilibrium} (see definition~\ref{def: mixedNE}).
This coincides with the notion of Nash equilibrium as defined in \citet{keskin2016equilibrium} and studied further in \citet{phade2019geometry}.
Now suppose each player can randomize over her actions and hence implement a black-box strategy.
If each player plays a fixed black-box strategy and these black-box strategies are common knowledge, then this gives rise to a new notion of equilibrium. 
We call it \emph{black-box strategy Nash equilibrium} (see definition~\ref{def: blackboxNE}).
If each player plays a fixed black-box strategy and the other players have mixed conjectures over her black-box strategy, and these conjectures are common knowledge, then we get the notion of \emph{mixed black-box strategy Nash equilibrium} (see definition~\ref{def: mixedblackboxNE}).

In the setting of an $n$-player normal form game with 
real valued 
payoff functions, the pure Nash equilibria do not depend on the specific CPT features of the players, i.e. the reference point, the value function and the two probability weighting functions, one for gains and one for losses.
Hence the traditional result on the lack of guarantee for the existence of a pure Nash equilibrium continues to hold when players have CPT preferences.
\citet{keskin2016equilibrium} proves the existence of a mixed action Nash equilibrium for any finite game when players have CPT preferences.
In example~\ref{ex: No blackboxNE}, we show that a finite game may not have any black-box strategy Nash equilibrium. 
On the other hand, in theorem~\ref{thm: BBNEexists}, we prove our main result that for any finite game with players having CPT preferences, there exists a mixed black-box strategy Nash equilibrium.
If the players have EUT preferences, then the notions of black-box strategy Nash equilibrium and mixed black-box strategy Nash equilibrium are equivalent to the notion of mixed action Nash equilibrium (when interpreted appropriately; see the remark before proposition~\ref{prop: compare}; see also figure~\ref{fig: venn_diagram}).

The paper is organized as follows. 
In section~\ref{sec: CPT_betweenness}, we describe the CPT setup and establish that under this setup betweenness is equivalent to independence (theorem~\ref{thm: CPTbet=indep}).
In section~\ref{sec: equilibrium}, we describe an $n$-player non-cooperative game setup and define various notions of Nash equilibrium in the absence of betweenness, in particular with CPT preferences.
We discuss the questions concerning their existence and how these different notions of equilibria compare with each other.
In section~\ref{sec: conclusion}, we conclude with a table that summarizes the results.

To close this section, we introduce some notational conventions that will be used in the document. 
If $Z$ is a Polish space 
(complete separable metric space), 
let $\cal{P}(Z)$ denote the set of all probability measures on $(Z, \cal{F})$, where $\cal{F}$ is the Borel sigma-algebra 
of $Z$. 
Let $\supp (p)$ denote the support of a distribution $p \in \cal{P}(Z)$, i.e. the smallest closed subset of $Z$ such that $p(Z) = 1$.
Let $\Delta_f(Z) \subset \cal{P}(Z)$ denote the set of all probability distributions that have a finite support. 
For any element $p \in \Delta_f(Z)$, let $p[z]$ denote the probability of $z \in Z$ assigned by $p$.
For $z \in Z$, let $\1\{z\} \in \Delta_f(Z)$ denote the probability distribution such that $p[z] = 1$.
If $Z$ is finite (and hence a Polish space with respect to the discrete topology), let $\Delta(Z)$ denote the set of all probability distributions on the set $Z$, viz.
\[
	\Delta(Z) = \cal{P}(Z) = \Delta_f(Z) = \l\{ (p[z])_{z \in Z} \bigg | p[z] \geq 0 \; \forall z \in Z, \sum_{z \in Z} p[z] = 1\r\},
\]
with the usual topology.
Let $\Delta^{m-1}$ denote the standard $(m-1)$-simplex, i.e. $\Delta(\{1, \dots, m\})$.
If $Z$ is a subset of a Euclidean space, then let $co(Z)$ denote the convex hull of $Z$, and let $\co(Z)$ denote the closed convex hull of $Z$.



\section{Cumulative Prospect Theory and Betweenness}
\label{sec: CPT_betweenness}

We first describe the setup for CPT (for more details see \citet{wakker2010prospect}). 
Each person is associated with a \emph{reference point} $r \in \bbR$, a \emph{value function} $v : \bbR \to \bbR$, and two \emph{probability weighting functions} $w^\pm:[0,1] \to [0,1]$, $w^+$ for gains and $w^-$ for losses. The function $v(x)$ satisfies:
\begin{inparaenum}[(i)]
	\item it is continuous in $x$,
	\item $v(r) = 0$,
    \item it is strictly increasing in $x$.
\end{inparaenum}
The value function is generally assumed to be convex in the losses frame ($x < r$) and concave in the gains frame ($x \geq r$), and to be steeper in the losses frame than in the gains frame in the sense that $v(r-z) \leq -v(r+z)$ for all $z \ge 0$. 
However, these assumptions are not needed for the results in this paper to hold.
The probability weighting functions $w^\pm: [0,1] \to [0,1]$ satisfy:
\begin{inparaenum}[(i)]
	\item they are continuous,
	\item they are strictly increasing,
	\item $w^\pm(0) = 0$ and $w^\pm(1) = 1$.
\end{inparaenum}
We say that $(r, v, w^\pm)$ are the \emph{CPT features} of that person.

Suppose a person faces a \emph{lottery} 
(or \emph{prospect}) $L := \{(p_k,z_k)\}_{1 \leq k \leq m}$,
where $z_k \in \bbR, 1 \leq k \leq m$, denotes an \emph{outcome} and $p_k, 1 \leq k \leq m$, is the probability with which outcome $z_k$ occurs. 
We assume that 
$\sum_{k=1}^m p_k = 1$. (Note that we are allowed to have $p_k =0$ for some values of $k$, and 
we can have 
$z_k = z_{k'}$ even when $k \neq k'$.) 
Let $z := (z_k)_{1 \leq k \leq m}$ and $p := (p_k)_{1 \leq k \leq m}$. We denote $L$ as $(p,z)$ and refer to the vector $z$ as an \emph{outcome profile} and $p$ as a \emph{probability vector}.

Let $\alpha := (\alpha_1,\dots,\alpha_m)$ be a permutation of $(1,\dots,m)$ such that
\begin{equation}\label{eq: order}
	z_{\alpha_1} \geq z_{\alpha_2} \geq \dots \geq z_{\alpha_m}.
\end{equation}
Let $0 \leq k_r \leq m$ be such that $z_{\alpha_k} \geq r$ for $1 \leq k \leq k_r$ and $z_{\alpha_k} < r$ for $k_r < k \leq m$. (Here $k_r = 0$ when $z_{\alpha_k} < r$ for all $1 \leq k \leq m$.) The \emph{CPT value} $V(L)$ of the prospect $L$ is evaluated using the value function $v(\cdot)$ and the probability weighting functions $w^{\pm}(\cdot)$ as follows:
\begin{equation}\label{eq: CPT_value_discrete}
	V(L) := \sum_{k=1}^{k_{r}} \pi_k^+(p,\alpha) v(z_{\alpha_k}) + \sum_{k=k_r+1}^m \pi_k^-(p,\alpha) v(z_{\alpha_k}),
\end{equation}
where 
$\pi^+_k(p,\alpha), 1 \leq k \leq k_{r},$ and $\pi^-_k(p,\alpha), k_r < k \leq m,$
are \emph{decision weights} defined via:
\begin{align*}
	\pi^+_{1}(p,\alpha) &:= w^+(p_{\alpha_1}),\\ 
	\pi_k^+(p,\alpha) &:= w^+(p_{\alpha_1} + \dots + p_{\alpha_{k}}) - w^+(p_{\alpha_1} + \dots + p_{\alpha_{k-1}}) &\text{ for } &1 < k \leq m, \\
	 \pi_k^-(p,\alpha) &:= w^-(p_{\alpha_m} + \dots + p_{\alpha_k}) - w^-(p_{\alpha_m} + \dots + p_{\alpha_{k+1}}) &\text{ for } &1 \leq k < m,\\
	 \pi^-_{m}(p,\alpha) &:= w^-(p_{\alpha_m}). 
\end{align*}
Although the expression on the right in equation~(\ref{eq: CPT_value_discrete}) depends on the permutation $\alpha$, one can check that the formula evaluates to the same value $V(L)$ as long as the permutation $\alpha$ satisfies (\ref{eq: order}). The CPT value in equation~(\ref{eq: CPT_value_discrete}) can equivalently be written as:
\begin{align}\label{eq: CPT_value_cumulative}
	V(L) &= \sum_{k = 1}^{k_r - 1} w^+\l(\sum_{i = 1}^k p_{\alpha_i}\r)\l[v(z_{\alpha_k}) - v(z_{\alpha_{k+1}})\r] \nonumber\\
	&+ w^+\l(\sum_{i = 1}^{k_r} p_{\alpha_i}\r)v\l(z_{\alpha_{k_r}}\r) + w^-\l(\sum_{i = k_r + 1}^{m} p_{\alpha_i}\r)v(z_{\alpha_{k_r+1}}) \nonumber \\
	&+ \sum_{k = k_r + 1}^{m-1} w^-\l(\sum_{i = k+1}^m p_{\alpha_i}\r)\l[v(z_{\alpha_{k+1}}) - v(z_{\alpha_{k}})\r].
\end{align}

A person is said to have CPT preferences if, given a choice between prospect $L_1$ and prospect $L_2$, she chooses the one with higher CPT value. 



We now define some axioms for preferences over lotteries.
We are interested in ``mixtures'' of lotteries, i.e. lotteries with other lotteries as outcomes.
Consider a (two-stage) compound lottery $K := \{(q^j, L^j)\}_{1 \leq j \leq t}$, where $L^j = (p^j, z^j), 1 \leq j \leq t$, are lotteries over real outcomes and $q^j$ is the chance of lottery $L^j$.
We assume that $\sum_{j= 1}^t q^j = 1$.
A two-stage compound lottery can be reduced to a single-stage lottery by multiplying the probability vector $p^j$ corresponding to the lottery $L^j$ by $q^j$ for each $j, 1 \leq j \leq t$, 
and then adding the probabilities of identical outcomes across all the lotteries $L^j, 1 \leq j \leq t$.
Let $\sum_{j = 1}^t q^j L^j$ denote the reduced lottery corresponding to the compound lottery $K$.

Let $\preceq$ denote a preference relation over single-stage lotteries.
We assume $\preceq$ to be a weak order, i.e. $\preceq$ is transitive (if $L_1 \preceq L_2$ and $L_2 \preceq L_3$, then $L_1 \preceq L_3$) and complete (for all 
$L_1, L_2$, we have $L_1 \preceq L_2$ or $L_2 \preceq L_1$, where possibly both preferences hold).
The additional binary relations $\succeq, \sim, \prec$ and $\succ$ are derived from $\preceq$ in the usual manner.
A preference relation $\preceq$ is a CPT preference relation if there exist CPT features $(r, v, w^\pm)$ such that $L_1 \preceq L_2$ iff $V(L_1) \leq V(L_2)$.
Note that a CPT preference relation is a weak order.
A preference relation $\preceq$ satisfies \emph{independence} if for any lotteries $L_1, L_2$ and $L$, and any constant $0 \leq \alpha \leq 1$, $L_1 \preceq L_2$ implies $\alpha L_1 + (1-\alpha)L \preceq \alpha L_2 + (1-\alpha) L$.
A preference relation $\preceq$ satisfies \emph{betweenness} if for any lotteries $L_1 \preceq L_2$, we have $L_1 \preceq \alpha L_1 + (1-\alpha)L_2 \preceq L_2$, for all $0 \leq \alpha \leq 1$.
A preference relation $\preceq$ satisfies \emph{weak betweenness} if for any lotteries $L_1 \sim L_2$, we have $L_1 \sim \alpha L_1 + (1 - \alpha) L_2$, for all $0 \leq \alpha \leq 1$.

Suppose a preference relation $\preceq$ satisfies independence.
Then $L_1 \preceq L_2$ implies 
$$L_1 = \alpha L_1 + (1 - \alpha) L_1 \preceq \alpha L_1 + (1 - \alpha) L_2 \preceq \alpha L_2 + (1-\alpha) L_2 = L_2.$$
Thus, if a preference relation satisfies independence, then it satisfies betweenness.
Also, if a preference relation satisfies betweenness, then it satifies weak betweenness.

In the following example, we will provide CPT features for Alice so that her preferences agree with those described in section~\ref{sec: intro}.
This example also shows that cumulative prospect theory can give rise to preferences that do not satisfy betweenness.
\begin{example}
\label{ex: Alice_nonbetweenness}
Recall that Alice is faced with the following three lotteries:
\begin{gather*}
	L_1 = \{( 0.34,\$20,000); (0.66,\$0)\},\\
	L_2 = \{( 0.17,\$30,000); (0.83,\$0)\},\\
	L = \{(0.01,\$30,000); (0.32,\$20,000); (0.67,\$0)\}.
\end{gather*}
Let $r = 0$ be the reference point of Alice.
Thus all the outcomes lie in the gains domain.
Let $v(x) = x^{0.8}$ for $x \geq 0$; Alice is risk-averse in the gains domain.
Let the probability weighting function for gains be given by
\[
		w^{+}(p) = \exp\{-(-\ln p)^{0.6}\},
\]
a form suggested by \citet{prelec1998probability} (see figure~\ref{fig: wt_pwf}).
We won't need the probability weighting function for losses. 
Direct computations show that $V(L_1) = 968.96, V(L_2) = 932.29$, and $V(L) = 1022.51$ (all decimal numbers in this example are correct to two decimal places).
Thus the preference behavior of Alice, as described in section~\ref{sec: intro} (i.e., she prefers $L_1$ over $L_2$, but prefers $L$ over $L_1$ and $L_2$), is consistent with CPT and can be modeled, for example, with the CPT features stated here.	
\qed
\end{example}

The following example shows that CPT can give rise to preferences that do not satisfy weak betweenness (the lotteries and the CPT features presented below appear in \citet{keskin2016equilibrium}).

\begin{example}
	\label{ex: non-monotone_lotteries}
	Suppose Charlie has $r = 0$ as his reference point and $v(x) = x$ as his value function. 
	Let his probability weighting function for gains be given by
	\[
		w^{+}(p) = \exp\{-(-\ln p)^{0.5}\}.
	\]
	(See figure~\ref{fig: wt_pwf}.)
	We won't need the probability weighting function for losses since we consider only outcomes in the gains domain in this example. 
	Consider the lotteries $L_1 = \{(0.5,2\beta);(0.5,0)\}$ and $L_2 = \{(0.5,\beta + 1);(0.5,1)\}$, where $\beta = 1/w^+(0.5) = 2.299$ (all decimal numbers in this example are correct to three decimal places).
	Direct computations reveal that $V(L_1) = V(L_2) = 2.000 > V(0.5L_1 + 0.5L_2)= 1.985$.

\qed
\end{example}


\input{plots/wt_plot_alice}

Given a utility function $u: \bbR \to \bbR$ 
(assumed to be continuous and strictly increasing), 
the expected utility of a lottery $L = \{(p_k, z_k)\}_{1 \leq k \leq m}$ is defined as $U(L) := \sum_{k = 1}^m p_k u(z_k)$.
A preference relation $\preceq$ is said to be an EUT preference relation if there exists a utility function $u$ such that $L_1 \preceq L_2$ iff $U(L_1) \leq U(L_2)$.
Note that if the CPT probability weighting functions are linear, i.e. $w^\pm(p) = p$ for $0 \leq p \leq 1$, then the CPT value of a lottery coincides with the expected utility of that lottery with respect to the utility function $u = v$.
It is well known that EUT preference relations satisfy independence and hence betweenness.
Several generalizations of EUT have been obtained by weakening the independence axiom and assuming only betweenness, for example, weighted utility theory \citep{chew1979alpha,hong1983generalization}, skew-symmetric bilinear utility \citep{fishburn1988nonlinear,bordley1991ssb}, implicit expected utility \citep{dekel1986axiomatic,chew1989axiomatic} and disappointment aversion theory \citep{gul1991theory,bordley1992intransitive}.
The following theorem shows that in the restricted setting of CPT preferences, betweenness and independence are equivalent.

\begin{theorem}
\label{thm: CPTbet=indep}
	If $\preceq$ is a CPT preference relation, then the following are equivalent:
	\begin{enumerate}[(i)]
		\item $\preceq$ is an EUT preference relation,
		\item $\preceq$ satisfies independence,	
		\item $\preceq$ satisfies betweenness.
	\end{enumerate}
\end{theorem}
\begin{proof}
Let the CPT preference relation $\preceq$ be given by $(r, v, w^\pm)$. 
Since an EUT preference relation satisfies independence, we get that (i) implies (ii).
Since betweenness is a weaker condition than independence, we get that (ii) implies (iii).
We will now show that if $\preceq$ satisfies betweenness, then the probability weighting functions are linear, i.e. $w^\pm(p) = p$ for $0 \leq p \leq 1$.
This will imply that $\preceq$ is an EUT preference relation with utility function $u = v$,
and hence complete the proof.

Assume that the CPT preference relation $\preceq$ satisfies betweenness.
Consider a lottery $A := \{(p_1, z_1), (p_2, z_2), (1 - p_1 - p_2, r)\}$ such that $z_1 \geq z_2 \geq r$, 
$p_1 \geq 0$, $p_2 > 0$ 
and $p_1 + p_2 \leq 1$.
By \eqref{eq: CPT_value_cumulative}, we have $$V(A) = \delta_1 w^+(P_1) + \delta_2 w^+(P_2),$$ where $\delta_1 := v(z_1) - v(z_2)$, $\delta_2 := v(z_2)$, $P_1 := p_1$ and $P_2 := p_1 + p_2$. 
Let lottery $B := \{(q_1, z_1), (q_2, z_2), (1 - q_1 - q_2, r)\}$ 
be such that $q_1,q_2 \geq 0$, $Q_1 := q_1 > p_1$, and $Q_2 := q_1 + q_2 < P_2$.
By \eqref{eq: CPT_value_cumulative}, we have $$V(B) = \delta_1 w^+(Q_1) + \delta_2 w^+(Q_2).$$
If $z_1, z_2, p_1, p_2, q_1$ and $q_2$ are such that
\begin{equation}
\label{eq: ratio_equal}
	\frac{\delta_1}{\delta_2} = \frac{w^+(P_2) - w^+(Q_2)}{w^+(Q_1) - w^+(P_1)},
\end{equation}
then $V(A) = V(B)$ and, by betweenness, for any $0 \leq \alpha \leq 1$ we have $V(A) = V(B) = V(\alpha A + (1-\alpha) B)$, i.e.
\begin{equation*}
	\delta_1 w^+(Q_1) + \delta_2 w^+(Q_2) = \delta_1 w^+(\alpha P_1 + (1-\alpha) Q_1) + \delta_2 w^+(\alpha P_2 + (1-\alpha)Q_2).
\end{equation*}
Using \eqref{eq: ratio_equal} we get
\begin{align}
\label{eq: ratio_weights_general}
	\l[w^+(P_2) - w^+(Q_2)\r]&\l[w^+(Q_1) - w^+(\alpha P_1 + (1-\alpha)Q_1))\r] \nonumber\\
	&= \l[w^+(Q_1) - w^+(P_1)\r]\l[w^+(\alpha P_2 + (1-\alpha)Q_2) - w^+(Q_2)\r].
\end{align}

Given any $0 \leq P_1 < Q_1 \leq Q_2 < P_2 \leq 1$, there exist $z_1$ and $z_2$ such that \eqref{eq: ratio_equal} holds. 
Indeed, take any $\delta > 0$ belonging to the range of the function $v$.
This exists because $v(r) = 0$ and $v$ is a strictly increasing function. 
Since $w^+$ is a strictly increasing function, we have
$$\kappa := \frac{w^+(P_2) - w^+(Q_2)}{w^+(Q_1) - w^+(P_1)} > 0.$$
Take 
$z_2 = v^{-1}({\delta}/{(1 + \kappa)})$ and $z_1 = v^{-1}(\delta)$.
These are well defined because $v$ is assumed to be continuous and strictly increasing, and $\delta$ belongs to its range.
Hence $z_1 > z_2 > r$ as required.
Thus \eqref{eq: ratio_weights_general} holds for any $0 \leq P_1 < Q_1 \leq Q_2 < P_2 \leq 1$. 
In particular,
when $Q_1 = Q_2$, we have
\begin{align*}
	\l[w^+(P_2) - w^+(Q)\r]&\l[w^+(Q) - w^+(R_1)\r] =  \l[w^+(Q) - w^+(P_1)\r]\l[w^+(R_2) - w^+(Q)\r],
\end{align*}	
where $Q := Q_1 = Q_2$, $R_1 := \alpha P_1 + (1-\alpha)Q$ and $R_2 := \alpha P_2 + (1-\alpha)Q$.
Equivalently, for any $0 \leq a_1 < c_1 < b < c_2 < a_2 \leq 1$ such that $(a_2 - b)(b - c_1) = (b - a_1)(c_2 - b)$, we have
 \begin{align*}
	\l[w^+(a_2) - w^+(b)\r]&\l[w^+(b) - w^+(c_1)\r] = \l[w^+(b) - w^+(a_1)\r]\l[w^+(c_2) - w^+(b)\r].
\end{align*}	
In lemma~\ref{lem: w_functionaleq}, we prove that the above condition implies $w^+(p) = p$, for $0 \leq p \leq 1$.
Similarly, we can show that $w^-(p) = p$, for $0 \leq p \leq 1$.
This completes the proof.
\end{proof}


\section{Equilibrium in black-box strategies}
\label{sec: equilibrium}

We now consider an $n$-player non-cooperative game where the players have CPT preferences. 
We will discuss several notions of equilibrium for such a game and will contrast them. 


Let $\Gamma := (N, (A_i)_{i \in N}, (x_i)_{i \in N})$ denote a \emph{game}, where $N := \{1, \dots, n\}$ is the set of \emph{players}, $A_i$ is the finite \emph{action set} of player $i$, and $x_i : A\to \bbR$ is the \emph{payoff function} for player $i$. 
Here $A := \prod_{i} A_i$ denotes the set of all \emph{action profiles} $a := (a_1, \dots, a_n)$.
Let 
$A_{-i} := \prod_{i \neq j} A_j$ 
denote the set of action profiles $a_{-i}$ of all players except player $i$.

\begin{definition}
\label{def: BAR_action_profile}
	For any action profile $a_{-i} \in A_{-i}$ of the opponents, we define the \emph{best response action set} of player $i$ to be
	\begin{equation}        \label{eq: pure-bestresponse}
	\BAR_i(a_{-i}) := \argmax_{a_i \in A_i}  x_i(a_i, a_{-i}).
\end{equation}
\end{definition}

\begin{definition}
\label{def: pureNE}
	An action profile $a = (a_1, \dots, a_n)$ is said to be a \emph{pure Nash equilibrium }if for each player $i \in N$, we have
	\[
		a_i \in \BAR_i(a_{-i}).
		\]
\end{definition}

The notion of pure Nash equilibrium is the same whether the players have CPT preferences or EUT preferences because only deterministic 
lotteries, comprised of being offered one outcome with probability $1$,
are considered in the framework of this notion. 
It is well known that for any given game $\Gamma$, a pure Nash equilibrium need not exist.

Let $\mu_{-i} \in \Delta(A_{-i})$ denote a \emph{belief} of player $i$ on the action profiles of her opponents.
Given 
the belief $\mu_{-i}$ of player $i$, 
if she decides to play action $a_i$, then 
she will face the lottery 
$\{(\mu_{-i}[a_{-i}], x_i(a_i, a_{-i}))\}_{a_{-i} \in A_{-i}}.$
\begin{definition}
\label{def: BAR}
	For any belief $\mu_{-i} \in \Delta(A_{-i})$, define the \emph{best response action set} of player $i$ as
 \begin{align}      \label{eq: cpt-bestresponse}
 \BAR_i(\mu_{-i}) := \argmax_{a_i \in A_i}  V_i\l(\l\{\l(\mu_{-i}[a_{-i}], x_i(a_i, a_{-i})\r)\r\}_{a_{-i} \in A_{-i}}\r).
 \end{align}
\end{definition}
Note that this definition is consistent with the definition of the best response action set that takes an action profile $a_{-i}$ of the opponents as its input (definition~\ref{def: BAR_action_profile}), if we interpret $a_{-i}$ as the belief $\1\{a_{-i}\} \in \Delta(A_{-i})$, since
$\BAR_i(\1\{a_{-i}\}) = \BAR_i(a_{-i})$.

Let $\sigma_i \in \Delta(A_i)$ denote a \emph{conjecture} over the action of player $i$. 
Let $\sigma := (\sigma_1, \dots, \sigma_n)$ denote a \emph{profile of conjectures}, and let $\sigma_{-i} := (\sigma_j)_{j \neq i}$ denote the profile of conjectures for all players except player $i$.
Let $\mu_{-i}(\sigma_{-i}) \in \Delta(A_{-i})$ be the belief induced by conjectures $\sigma_j, j \neq i$, given by
\[
	\mu_{-i}(\sigma_{-i})[a_{-i}] := \prod_{j \neq i} \sigma_j[a_{-i}],
\]
which is nothing but the product distribution induced by $\sigma_{-i}$.

\begin{definition}
\label{def: mixedNE}
	A conjecture profile $\sigma = (\sigma_1, \dots, \sigma_n)$ is said to be a \emph{mixed action Nash equilibrium} if, for each player $i$, we have
	\[
		a_i \in \BAR_i(\mu_{-i}(\sigma_{-i})), \text{ for all } a_i \in \supp \sigma_i.
	\]
\end{definition}

In other words, the conjecture $\sigma_i$ over the action of player $i$ should assign positive probabilities to only optimal actions of player $i$, given her belief $\mu_{-i}(\sigma_{-i})$.


It is well known that a mixed Nash equilibrium exists for every 
game with EUT players, see \citet{nash1951non}.
\citet{keskin2016equilibrium} generalizes the result of \citet{nash1951non} on the existence of a mixed action Nash equilibrium to the case when players have CPT preferences.

Let $B_i := \Delta(A_i)$ denote the set of all black-box strategies for player $i$ with a typical element denoted by $b_i \in B_i$. 
Recall that if player $i$ implements a black-box strategy $b_i$, then we interpret this as a trusted party other than the player sampling an action $a_i \in A_i$ from the distribution $b_i$ and playing action $a_i$ on behalf of player $i$.
We assume the usual topology on $B_i$.
Let $B := \prod_{i} B_i$ and $B_{-i} := \prod_{j \neq i} B_j$ with typical elements denoted by $b$ and $b_{-i}$, respectively.

Note that, although a conjecture $\sigma_i$ and a black-box strategy $b_i$ are mathematically equivalent, viz. they are elements of the same set $B_i = \Delta(A_i)$, they have different interpretations.
We will call $s_i \in \Delta(A_i)$ a \emph{mixture} of actions of player $i$ 
when we want to be agnostic to which interpretation is being 
imposed. 
Let $S_i := \Delta(A_i), S := \prod_i \Delta(A_i)$ and $S_{-i} := \prod_{j \neq i} S_i$ with typical elements denoted by $s_i, s$ and $s_{-i}$, respectively.
(Note that $S \neq \Delta(A)$ unless all but one player have singleton action sets.)

For any belief $\mu_{-i} \in \Delta(A_{-i})$ and any black-box strategy $b_i$ of player $i$, 
let $\mu(b_i, \mu_{-i}) \in \Delta(A)$ denote the product distribution given by
\[
	\mu(b_i, \mu_{-i})[a] := b_i[a_i] \mu_{-i}[a_{-i}].
\]
Given 
the belief $\mu_{-i}$ of player $i$, 
if she decides to implement the black-box strategy $b_i$, then she 
will 
face the lottery 
$\{\mu(b_i,\mu_{-i})[a], x_i(a))\}_{a \in A}$.
\begin{definition}
\label{def: BBR}
	For any belief $\mu_{-i} \in \Delta(A_{-i})$,
define the \emph{best response black-box strategy set} of player $i$ as
\begin{align*}
\BBR_i(\mu_{-i}) := \argmax_{b_i \in B_i} V_i\l(\l\{\l(\mu(b_i,\mu_{-i})[a], x_i(a)\r)\r\}_{a \in A}\r).
\end{align*}
\end{definition}

\begin{lemma}
\label{lem: BBR_closed}
	For any belief $\mu_{-i}$, the set $\BBR_i(\mu_{-i})$ is non-empty, and
	\[
		\co(\BBR_i(\mu_{-i})) = co(\BBR_i(\mu_{-i})).
	\]
\end{lemma}
\begin{proof}
	For a lottery $L = (p,z)$, where $z = (z_k)_{1 \leq k \leq m}$ is the outcome profile, and $(p_k)_{1 \leq k \leq m}$ is the probability vector, the function $V_i(p,z)$ is continuous with respect to $p \in \Delta^{m-1}$
	 \citep{keskin2016equilibrium}.
	Thus,
	$V_i(\{(\mu(b_i, \mu_{-i})[a], x_i(a))\}_{a \in A})$ 
	is a function continuous with respect to $b_i \in B_i$,
	and hence
	$\BBR_i(\mu_{-i})$ is a non-empty closed subset
	of the compact space $B_i$. 
	Since the convex hull of a compact subset of a Euclidean space is compact \citep[Chapter~3]{rudin1991functional}, the set $co(\BBR_i(\mu_{-i}))$ is closed.
	This completes the proof. 
\end{proof}

Let us compare the two concepts: the best response action set (definition~\ref{def: BAR}) and the best response black-box strategy set (definition~\ref{def: BBR}). 
Even though both of them take the belief $\mu_{-i}$ of player $i$ as input,
the best response action set $\BAR_i(\mu_{-i})$ outputs a collection of actions of player $i$, 
whereas the best response black-box strategy set $\BBR_i(\mu_{-i})$ outputs a collection of black-box strategies of player $i$, which are probability distributions over the set of actions $a_i \in A_i$.
If we interpret an action $a_i$ as the mixture 
$\1\{a_i\} \in S_i = \Delta(A_i)$, 
and a black-box strategy $b_i$ as a mixture as well, then we can compare the two sets $\BAR(\mu_{-i})$ and $\BBR(\mu_{-i})$ as subsets of $S_i$.
The following example shows that, in general, the two sets can be disjoint, and hence quite distinct.

\begin{example}
	\label{ex: Alice_game}
	We consider a $2$-player game. Let Alice be player $1$, with action set $A_1 = \{1,2\}$, and let Bob be player $2$, with action set $A_2 = \{1,2,3\}$.
	Let the payoff function for Alice be as shown in figure~\ref{tab: Alice_payoff}.
	Let $\mu_{-1} = (0.17, 0.17, 0.66) \in \Delta(A_{-1}) = \Delta(A_2)$ be the belief of Alice.
	Then, as considered in section~\ref{sec: intro}, Alice faces the lottery
	$L_1 = \{( 0.34,\$20,000); (0.66,\$0)\}$
	if she plays action $1$ and
	the lottery
	$L_2 = \{( 0.17,\$30,000); (0.83,\$0)\}$
	if she plays action $2$.
	We retain the CPT features for Alice, as in example~\ref{ex: Alice_nonbetweenness}, viz.: $r = 0$, $v(x) = x^{0.8}$ for $x \geq 0$, and 
	\[
		w^{+}(p) = \exp\{-(-\ln p)^{0.6}\}.
	\]
	We saw that $V_1(L_1) = 968.96$, $V_1(L_2) = 932.29$, and $V(16/17 L_1 + 1/17 L_2) = 1022.51$ (all decimal numbers in this example are correct to two decimal places).
	Amongst all the mixtures, the maximum CPT value is achieved at the unique mixture $L^* = \alpha^* L_1 + (1-\alpha^*) L_2$, where $\alpha^* = 0.96$; we have $V_1(L^*) = 1023.16$.
	Thus, $\BAR_1(\mu_{-1}) = \{\1\{1\}\}$ and $\BBR_1(\mu_{-1}) = \{ (\alpha^*, 1 -\alpha^*) \}$.
\begin{figure}
\centering
\begin{tabular}{c | c | c | c |}
	 \multicolumn{1}{c}{}	& \multicolumn{1}{c}{$1$}   &  \multicolumn{1}{c}{$2$} &  \multicolumn{1}{c}{$3$}\\
	 \cline{2-4}
	 $1$ & $\$20{,}000$ & $\$20{,}000$  &  $\$0$\\
	 \cline{2-4}
	 $2$	& $\$30{,}000$ & $\$0$  & $\$0$ \\
	 \cline{2-4}
	 \end{tabular}
	 \caption{Payoff matrix for Alice in example~\ref{ex: Alice_game}. Rows and columns correspond to Alice's and Bob's actions respectively. The amount in each cell corresponds to Alice's payoff.}
	 \label{tab: Alice_payoff}
\end{figure}
\qed
\end{example}


For any black-box strategy profile $b_{-i}$ of the opponents,
let $\mu_{-i}(b_{-i}) \in \Delta(A_{-i})$ be the induced belief given by
\[
	\mu_{-i}(b_{-i})[a_{-i}] := \prod_{j \neq i} b_j[a_{-i}].
\]

\begin{definition}
\label{def: blackboxNE}
	A black-box strategy profile $b = (b_1, \dots, b_n)$ is said to be a \emph{black-box strategy Nash equilibrium} if, for each player $i$, we have
	\[
		b_i \in \BBR_i(\mu_{-i}(b_{-i})).
	\]
\end{definition}

If the players have EUT preferences, 
a conjecture profile $\sigma = (\sigma_1, \dots, \sigma_n)$ is a mixed action Nash equilibrium 
if and only if the black-box strategy profile $b = (b_1, \dots, b_n)$, where $b_i = \sigma_i$, for all $i \in N$, 
is a black-box strategy Nash equilibrium.
Thus, under EUT, the notion of a black-box strategy Nash equilibrium is equivalent to the notion of a mixed action Nash equilibrium,
although there is still a conceptual difference between these two notions based on the interpretations for the mixtures of actions.
Further, we have the existence of a black-box strategy Nash equilibrium for any game when players have EUT preferences from the 
well-known result about the existence of a mixed action Nash equilibrium.
The following example shows that, in general, a black-box strategy Nash equilibrium may not exist when players have CPT preferences.

\begin{example}
\label{ex: No blackboxNE}
Consider a $2\times2$ game (i.e a $2$-player game where each player has two actions $\{0,1\}$) with the payoff matrices as shown in figure~\ref{tab: 2x2 game}.
Let the reference points be $r_1 = r_2 = 0$. 
Let $v_i(\cdot)$ be 
the identity function
for $i = 1,2$.
Let the probability weighting functions for gains for the two players be given by
\[
	w_i^+(p) = \exp \{-(-\ln p)^{\gamma_i}\}, \text{ for } i = 1,2,
\]
where $\gamma_1 = 0.5$ and $\gamma_2 = 1$.
We do not need the probability weighting functions for losses since all the outcomes lie in the gains domain for both the players.
Notice that player $2$ has EUT preferences since $w_2^+(p) = p$.

\begin{figure}
\parbox{0.45\linewidth}{
\centering
\begin{tabular}{c | c | c |}
	 \multicolumn{1}{c}{}	& \multicolumn{1}{c}{0}   &  \multicolumn{1}{c}{1}\\
	 \cline{2-3}
	 0 & $4$  &  $0$\\
	 \cline{2-3}
	 1	& $3$ & $1$ \\
	 \cline{2-3}
	 \end{tabular}}
\hfill
\parbox{0.45\linewidth}{
\centering
\begin{tabular}{c | c | c |}
	 \multicolumn{1}{c}{}	& \multicolumn{1}{c}{0}   &  \multicolumn{1}{c}{1}\\
	 \cline{2-3}
	 0 & $0$  &  $1$\\
	 \cline{2-3}
	 1	& $1$ & $0$ \\
	 \cline{2-3}
	 \end{tabular}
}
	 \caption{Payoff matrices for the
	 $2\times2$ game
	 in example~\ref{ex: No blackboxNE} (left matrix for player $1$ and right matrix for player $2$). 
	 The rows and the columns correspond to the actions of player $1$ and player $2$, respectively, and the entries in the cell represent the corresponding payoffs.
     }\label{tab: 2x2 game}
\end{figure}

Suppose player $1$ and player $2$ play black-box strategies $(1-p, p)$ and $(1-q, q)$, respectively, where $p, q \in [0,1]$.
With an abuse of notation, 
we identify these black-box strategies by $p$ and $q$, respectively.
The corresponding lottery faced by player $1$ is given by
\[
	L_1(p,q) := \{(\mu[0,0], 4); (\mu[1,0], 3); (\mu[1,1], 1); (\mu[0,1], 0)\},
\]
where $\mu[0,0] := (1-p)(1-q), \mu[1,0] := p(1-q), \mu[0,1] := (1-p)q$, and $\mu[1,1] := pq$.
By \eqref{eq: CPT_value_discrete}, the CPT value of the lottery faced by player $1$ is given by
\begin{align*}
	V_1(L_1(p,q)) &:= 4 \times \l[w^+_1(\mu[0,0])\r]\\
	& + 3 \times \l[w^+_1(\mu[0,0] + \mu[1,0]) - w^+_1(\mu[0,0]))\r] \\
	&+ 1 \times \l[w^+_1(\mu[0,0] + \mu[1,0] + \mu[1,1]) - w^+_1(\mu[0,0] + \mu[1,0])\r].
\end{align*}
The plot of the function $V_1(L_1(p,q))$ with respect to $p$, for $q = 0.3$ and $q = 0.35$, is shown in figure~\ref{fig: f_plot}.
We observe that the best response black-box strategy set $\BBR_1(\mu_{-1}(q))$ of player $1$ to player $2$'s black-box strategy $q \in B_2$ satisfies the following: 
$\BBR_1(\mu_{-1}(q)) = \{0\}$ for $q < q^*$, 
$\BBR_1(\mu_{-1}(q)) = \{0, p^*\}$ for $q = q^*$, and 
$\BBR_1(\mu_{-1}(q)) \subset [p^*, 1]$ for $q > q^*$,
where $p^* = 0.996$ and $q^* = 0.340$ (here the numbers are correct to three decimal points).
Further, $\BBR_1(\mu_{-1}(q))$ is singleton for $q \in (q^*, 1]$ and the unique element in $\BBR_1(\mu_{-1}(q))$ increases monotonically with respect to $q$ from $p^*$ to $1$ (see figure~\ref{fig: noBBNE}).
In particular, $\BBR_1(\mu_{-1}(1)) = \{1\}$.
The lottery faced by player $2$ is given by 
\[
	L_2(p,q) := \{(\mu[0,0], 0); (\mu[1,0], 1); (\mu[1,1], 0); (\mu[0,1], 1)\},
\]
and the CPT value of player $2$ for this lottery is given by
$V_2(L_2(p,q)) = p(1-q) + q(1-p)$.
The best response black-box strategy set 
$\BBR_2(\mu_{-2}(p))$
of player $2$ to player $1$'s black-box strategy $p \in B_1$ satisfies the following: 
$\BBR_2(\mu_{-2}(p)) = \{1\}$ for $p < 0.5$, 
$\BBR_2(\mu_{-2}(p)) = [0,1]$ for $p = 0.5$, and 
$\BBR_2(\mu_{-2}(p)) = \{0\}$ for $p > 0.5$.
As a result, see figure~\ref{fig: noBBNE}, there does not exist any $(p', q')$ such that $p' \in \BBR_1(\mu_{-1}(q'))$ and $q' \in \BBR_2(\mu_{-2}(p'))$, and hence no black-box strategy Nash equilibrium exists for this game.

\input{plots/Value_plot_example_noBBNE}


\begin{figure}
\centering
	\includegraphics[scale = 0.3]{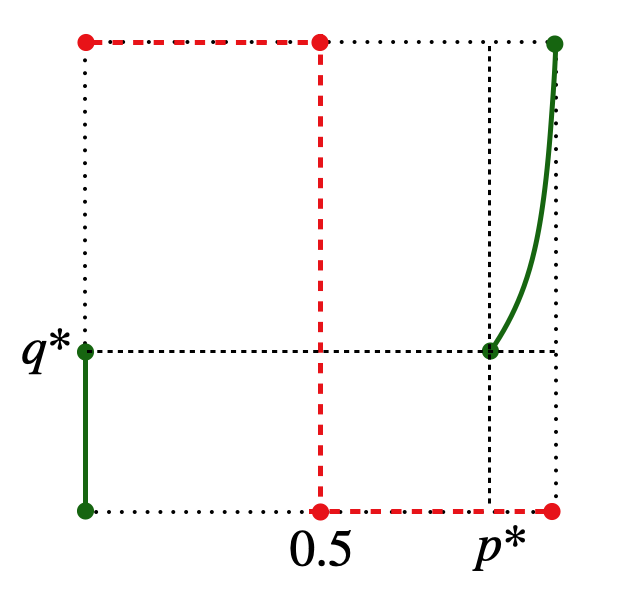}
	\caption{The figure (not to scale) shows the best response black-box strategy sets of the two players for the game in example~\ref{ex: No blackboxNE}. The red (dashed) line shows the best response black-box strategy set of player $2$ in response to the black-box strategy $(1-p, p)$ of player $1$.
	The green (solid) line shows the best response black-box strategy set of player $1$ in response to the black-box strategy $(1-q, q)$ of player $2$. Note that there is no intersection of these lines.}
	\label{fig: noBBNE}
\end{figure}


\qed
\end{example}

Let $\tau_i \in \cal{P}(B_i)$ denote a conjecture over the black-box strategy of player $i$.
This will induce a conjecture $\sigma_i(\tau_i) \in \Delta(A_i)$ over the action of player $i$, given by 
\[
	\sigma_i(\tau_i)[a_i] = \bbE_{\tau_i} b_i[a_i].
\]
Given conjectures over black-box strategies 
$(\tau_j \in \Delta(B_j), j \neq i)$,
let $\sigma_{-i}(\tau_{-i}) := (\sigma_j(\tau_j))_{j \neq i}$. \cob

\begin{definition}
\label{def: mixedblackboxNE}
	A profile of conjectures over black-box strategies $\tau = (\tau_1, \dots, \tau_n)$ is said to be a \emph{mixed black-box strategy Nash equilibrium} 
	if, for each player $i$, 
	we have
	\[
		b_i \in \BBR_i(\mu_{-i}(\sigma_{-i}(\tau_{-i}))), \text{ for all } b_i \in \supp \tau_i.
	\]
\end{definition}

\begin{proposition}
\label{prop: sigma_exist_eq}
For a profile of conjectures $\sigma^* = (\sigma_1^*, \dots, \sigma_n^*)$, consider the condition
\begin{equation}
	\label{eq: sigma_eq_cond}
		\sigma_i^* \in \co(\BBR_i(\mu_{-i}(\sigma_{-i}^*))), \text{ for all } i.
\end{equation}
\begin{enumerate}[(i)]
	\item If $\tau$ is a mixed black-box strategy Nash equilibrium, then the profile of conjectures $\sigma^*$, where $\sigma_i^* = \sigma_i(\tau_i), \forall i$, satisfies \eqref{eq: sigma_eq_cond}.
	\item If $\sigma^*$ satisfies \eqref{eq: sigma_eq_cond}, then there exists a profile of finite support conjectures on black-box strategies $\dot \tau = (\dot \tau_1, \dots, \dot \tau_n)$, where $\dot \tau_i \in \Delta_f(B_i), \forall i$, that is a mixed black-box strategy Nash equilibrium, such that $\sigma_i^* = \sigma_i(\dot \tau_i),\forall i$.
\end{enumerate}
	
	
\end{proposition}
\begin{proof}
	Suppose $\tau$ is a mixed black-box strategy Nash equilibrium.
	Let $\sigma^*_i = \sigma_i(\tau_i)$.
	Then, for all $b_i \in \supp \tau_i$, we have $b_i \in \BBR_i(\mu_{-i}(\sigma^*_{-i}))$,
	and hence $\sigma^*_i \in \co(\BBR_i(\mu_{-i}(\sigma^*_{-i})))$.
	This proves statement (i).

	For statement (ii), suppose $\sigma^*$ satisfies condition \eqref{eq: sigma_eq_cond}. 
	In fact, by lemma~\ref{lem: BBR_closed} we have, $\sigma^*_i \in co(\BBR_i(\mu_{-i}(\sigma^*_{-i}))) \subset \Delta(A_i)$, and by Caratheodory's theorem, $\sigma^*_i$ is a convex combination of at most $|A_i|$ elements in $\BBR_i(\mu_{-i}(\sigma^*_{-i}))$.
	Hence, we can construct a mixed black-box strategy Nash equilibrium $\dot \tau$ such that $\dot \tau_i \in \Delta_f(B_i)$ and $\sigma_i^* = \sigma_i(\dot \tau_i), \forall i$.
\end{proof}

The content of this proposition is that in order to determine whether a
profile $\tau$ of conjectures on black box strategies is a mixed black-box strategy Nash equilibrium or not it suffices to study the associated profile of
conjectures on actions that is induced by $\tau$. This justifies the study of the set $\mBBNE$ discussed below.

\begin{theorem}
\label{thm: BBNEexists}
	For any game $\Gamma$, there exists a profile of conjectures $\sigma^* = (\sigma_1^*, \dots, \sigma_n^*)$ that satisfies \eqref{eq: sigma_eq_cond}.
\end{theorem}
\begin{proof}
The idea is to use the Kakutani fixed-point theorem, as in the proof of the existence of mixed action Nash equilibrium \citep{nash1950equilibrium}. 
Assume the usual topology on $S_i$, for each $i$, and let $S$ have the corresponding product topology.
The set $S$ is a non-empty compact convex subset of the Euclidean space $\prod_{i} \bbR^{|A_i|}$.
Let $K(\sigma)$ be the set-valued function given by
\[
	K(\sigma) := \prod_i \co(\BBR_i(\mu_{-i}(\sigma_{-i}))),
\]
for all $\sigma \in S$.
Since $\co(\BBR_i(\mu_{-i}(\sigma_{-i})))$ is non-empty and convex for each $i$ (lemma~\ref{lem: BBR_closed}), 
the function $K(\sigma)$ is non-empty and convex for any $\sigma \in S$. 
We now show that the function $K(\cdot)$ has a closed graph. 
Let $\{\sigma^t\}_{t=1}^\infty$ and $\{s^t\}_{t=1}^\infty$ be 
two sequences in $S$ that converge to
$\bar \sigma$ and $\bar s$, respectively, 
and let $s^t \in K(\sigma^t)$ for all $t$. 
It is enough to show that $\bar s \in K(\bar \sigma)$. 
For all 
$s_i \in S_i, \sigma_{-i} \in S_{-i}$, 
let 
\begin{equation*}
	\tilde V_i(s_i,\sigma_{-i}) := \sup_{\substack{\tau_i \in \cal{P}(B_i),\\ 
	\bbE_{\tau_i} b_i = s_i
	}} 
	\bbE_{\tau_i}
	V_i\l(\{(\mu(b_i,\mu_{-i}(\sigma_{-i}))[a],x_i(a))\}_{a \in A}\r).
\end{equation*}
Since the product distribution $\mu(b_i,\mu_{-i}(\sigma_{-i}))$ is jointly continuous in $b_i$ and $\sigma_{-i}$, 
and, as noted earlier, $V_i(p,z)$ is continuous with respect to the probability vector $p$, for any fixed outcome profile $z$, the function $V_i\l(\{\mu(b_i,\mu_{-i}(\sigma_{-i}))[a],x_i(a)\}_{a \in A}\r)$ is jointly continuous in $b_i$ and $\sigma_{-i}$. 
This implies that the function 
$\tilde V_i(s_i,\sigma_{-i})$ 
is jointly continuous in 
$s_i$
and $\sigma_{-i}$ (see Appendix \ref{app: continuity}).
From the definition of $\tilde V_i$, it follows that 
$$\max_{s_i \in \Delta(A_i)} \tilde V_i(s_i, \sigma_{-i}) =  \max_{b_i \in B_i}  V_i\l(\{(\mu(b_i,\mu_{-i}(\sigma_{-i}))[a],x_i(a))\}_{a \in A}\r).$$
Indeed, the maximum on the left-hand side is well-defined since $\Delta(A_i)$ is a compact space and $\tilde V_i(\cdot, \sigma_{-i})$ is a continuous function.
The maximum on the right-hand side is well-defined and the maximum is achieved by all $b_i \in \BBR_i(\mu_{-i}(\sigma_{-i}))$ (lemma~\ref{lem: BBR_closed}).
Hence,
$$\argmax_{s_i \in \Delta(A_i)} \tilde V_i(s_i,\sigma_{-i}) = \co(\BBR_i(\mu_{-i}(\sigma_{-i}))).$$
Since $s_i^t \in \co(\BBR_i(\mu_{-i}(\bar \sigma_{-i}^t)))$, for all $t$, we have
$$\tilde V_i(s_i^t,\sigma^t_{-i}) \geq \tilde V_i(\tilde s_i,\sigma^t_{-i}),~~\mbox{for all $\tilde s_i \in S_i$.}$$
Since 
$\tilde V_i(s_i,\sigma_{-i})$ 
is jointly continuous in 
$s_i$ 
and $\sigma_{-i}$, we get 
$$\tilde V_i(\bar s_i, \bar \sigma_{-i}) \geq \tilde V_i(\tilde s_i,\bar \sigma_{-i}),~~\mbox{for all $\tilde s_i \in S_i$.}$$
Hence we have $\bar s_i \in \co(\BBR_i(\mu_{-i}(\bar \sigma_{-i}^t)))$.
This shows that the function $K(\cdot)$ has a closed graph.
By the Kakutani fixed-point theorem, there exists $\sigma^*$ such that $ \sigma^* \in K(\sigma^*)$, i.e. $\sigma^*$ satisfies condition \eqref{eq: sigma_eq_cond} \citep{kakutani1941generalization}.
This completes the proof.
\end{proof}

\begin{corollary}
	For any finite game $\Gamma$, there exists a mixed black-box strategy Nash equilibrium.
	In particular, there is one that is a profile of finite support conjectures over the black-box strategies of players.
\end{corollary}
\begin{proof}
	Follows from theorem~\ref{thm: BBNEexists} and statement (ii) of proposition~\ref{prop: sigma_exist_eq}.
\end{proof}



We now compare the different notions of Nash 
equilibrium
defined above.
To that end, we will associate each of the equilibrium notions with their corresponding natural profile of mixtures over actions. 
For example, corresponding to any pure Nash equilibrium $a = (a_1, \dots, a_n)$, assign the profile of mixtures over actions $(\1\{a_1\}, \dots, \1\{a_n\}) \in S$.
Let $\pNE \subset S$ denote the set of all profiles of mixtures over actions that correspond to pure Nash equilibria.
Let $\mNE \subset S$ denote the set of all mixed action Nash equilibria $\sigma \in S$.
Let $\BBNE \subset S$ denote the set of all black-box strategy Nash equilibria $b \in S$.
Corresponding to any mixed black-box strategy Nash equilibrium $\tau = (\tau_1, \dots, \tau_n)$, assign the profiles of mixtures over actions $(\sigma_1(\tau_1), \dots, \sigma_n(\tau_n)) \in S$, and let $\mBBNE \subset S$ denote the set of all such profiles.
Note that each of the above subsets 
depends
on the underlying game $\Gamma$ and the CPT features of the players.

\begin{proposition}
\label{prop: compare}
	For any fixed game $\Gamma$ and CPT features of the players, we have
\begin{enumerate}[(i)]
 	\item $\pNE \subset \mNE,$
 	\item $\pNE \subset \BBNE,$ and
 	\item $\BBNE \subset \mBBNE.$
 \end{enumerate} 
\end{proposition}
\begin{proof}
	The proof of statement (i) can be found in \citet{keskin2016equilibrium}.

	For statement (ii), let $(\1\{a_1\}, \dots, \1\{a_n\}) \in \pNE$.
	For a black-box strategy $b_i$ of player $i$, the belief $\mu_{-i} = \1\{a_{-i}\}$ of player $i$ gives rise to the lottery 
	$\{(b_i[a_i'], x_i(a_i', a_{-i}))\}_{a_i' \in A_i}$.
	From the definition of CPT value (see equation~\eqref{eq: CPT_value_discrete}), we observe that $V_i(\{(b_i[a_i'], x_i(a_i', a_{-i}))\}_{a_i' \in A_i})$ is optimal as long as the probability distribution $b_i$ does not assign positive probability to any suboptimal outcome.
	Hence,
	\[
		\BBR_i(\1\{a_{-i}\}) = \co(\1\{a_i'\} \in S_i : a_i' \in \BAR_i(\1\{a_{-i}\})). 
	\]
	In particular, $\1\{a_i\} \in \BBR_i(\1\{a_{-i}\})$, and hence $(\1\{a_1\}, \dots, \1\{a_n\}) \in \BBNE$.

	Statement (iii) follows directly from the definitions~\ref{def: blackboxNE} and \ref{def: mixedblackboxNE}.
\end{proof}

\begin{figure}
\centering
    \begin{tikzpicture}
    \draw[thick] (-5,-4) rectangle (5,4) node [label={[shift={(-9,-1)}]$S$}] {};
      \node (A) at (-0.6,3) {$\pNE$};
      \draw[very thick] (0,0) circle (0.7) node (B) [shift={(-0.1,0.6)}] {};
      \draw[thick,->] (A) -- (B);
      \draw[very thick] (-1.2,-0.3) circle (2.2) node [label={[shift={(-1.5,1.8)}]$\mNE$}] {};
      \node (C) at (1.6,3.5) {$\BBNE$};
      \draw[very thick] (0.6,-0.3) circle (1.8) node (D) [shift={(0.6,1.7)}] {};
      \draw[thick,->] (C) -- (D);
      \node (C) at (1.6,3.5) {$\BBNE$};
      \draw[very thick] (1,-0.2) circle (3) node [label={[shift={(3.1,1.8)}]$\mBBNE$}] {};
      \draw (0,0) circle (0.2) node [text=black] {a};
      \draw (-2.7,-0.2) circle (0.2) node [text=black] {b};
      \draw (-0.2,-1.2) circle (0.2) node [text=black] {c};
      \draw (-1.2,1) circle (0.2) node [text=black] {d};
      \draw (1.6,-0.3) circle (0.2) node [text=black] {e};
      \draw (3,0.7) circle (0.2) node [text=black] {f};
      \draw (-4,0.7) circle (0.2) node [text=black] {g};
    \end{tikzpicture}
    \caption{Venn diagram depicting the different notions of 
    equilibrium
    as subsets of the set $S = \prod_i \Delta(A_i)$.
    The sets marked $\pNE, \mNE, \BBNE$, and $\mBBNE$ represent the sets of pure Nash equilibria, mixed action Nash equilibria, black-box strategy Nash equilibria, and mixed black-box strategy Nash equilibria, respectively.
    Examples are given in the body of the text of CPT games lying in each of the indicated regions (a) through (g).
    }
    \label{fig: venn_diagram}
\end{figure}
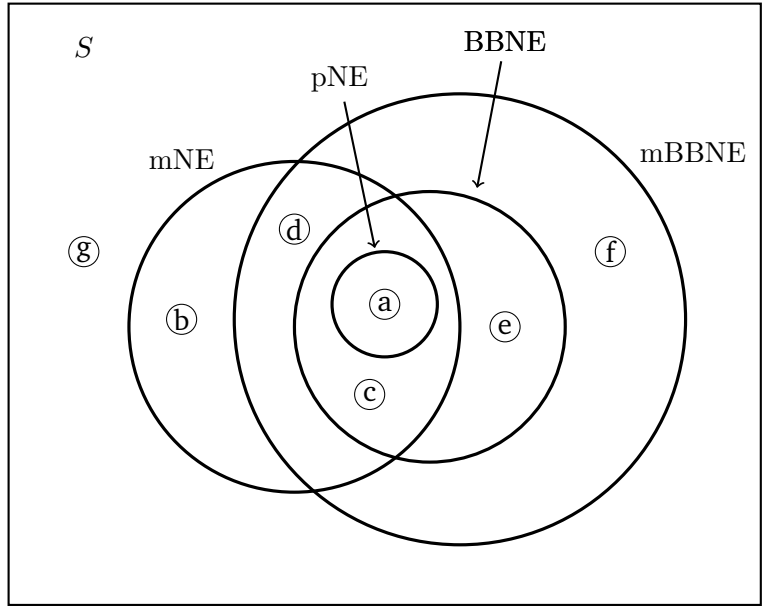

In the following, we show via examples
that each of the labeled regions ((a)--(g)), in figure~\ref{fig: venn_diagram}, is non-empty in general.

\begin{example}
\label{ex: non-empty_regions}
For each of the seven regions in figure~\ref{fig: venn_diagram}, we provide a $2 \times 2$ game with the accompanying CPT features for the two players 
verifying that the corresponding region is non-empty.
Let the action sets be $A_1 = A_2 = \{0, 1\}$.
With an abuse of notation, let $p, q \in [0,1]$ denote the mixtures over actions for players $1$ and $2$, respectively, where $p$ and $q$ are the probabilities corresponding to action $1$ for both the players.
Thus, the set of all profiles of mixtures over actions is $S = \{(p,q) : p, q \in [0,1]\}$.
Let $L_1(p,q)$ and $L_2(p,q)$ denote the corresponding lotteries faced by the two players.
(All decimal numbers in 
these examples
are correct to three decimal places.)

\begin{enumerate}[(a)]
	\item Let both the players have EUT preferences with their utility functions given by the identity functions $u_i(x) = x$, for $i = 1, 2$.
	Let the payoff matrix be as shown in figure~\ref{subfig: a}.
	Clearly, $(p=0, q=0) \in \pNE$.

	\item Let $r_i = 0, v_i(x) = x$, for $i = 1, 2$.
	Let
	$w_1^{+}(p) = p^{0.5}$ and $w_2^+(p) = p$, for $p \in [0,1]$.
	Let the payoff matrix be as shown in figure~\ref{subfig: b}, where $\beta := 1/w_1^+(0.5) = 1.414$.
	We have
	\[
		L_1(p,q) = \{((1-p)(1-q), 2\beta); (p(1-q), \beta + 1); (pq, 1); ((1-p)q, 0)\}.
	\]
	The way $\beta$ is defined, we get $V_1(L_1(0, 0.5)) = V_1(L_1(1,0.5)) = 2$.
	Also, observe that $V_2(L_2(0.5,0)) = V_2(L_2(0.5, q))  = V_2(L_2(0.5,1)), \forall q \in [0,1].$
	With these observations, we get that $(0.5, 0.5) \in \mNE$.
	We have, $\argmax_{p \in [0,1]} V_1(L_1(p, 0.5)) = \{p'\}$, where $p' = 0.707$ 
	(see figure~\ref{plot: region_b}).
	Hence $0.5 \notin \co(\BBR_1(\mu_{-1}(0.5)))$ and $(0.5, 0.5) \notin \mBBNE$.

	\item Let the CPT features for both the players be as in (b).
	Let the payoff matrix be as shown in figure~\ref{subfig: c}, where $\beta := 1/w_1^+(0.5) = 1.414$ and $\gamma = (1-p')/p'$ (here $p' = 0.707$ as in (b)).
	As observed in (b), $\BBR_1(\mu_{-1}(0.5)) = \{p'\}$.
	From the definition of $\gamma$, we see that player $2$ is indifferent between her two actions, given her belief $p'$ over player $1$'s actions.
	Thus $(p', 0.5) \in (\mNE \cap \BBNE) \back \pNE$.

	\item Let $r_i = 0, v_i(x) = x$, for $i =1,2$.
	Let
	$w_1^{-}(p) = p^{0.5}, w_2^+(p) = p$.
	Let the payoff matrix be as shown in figure~\ref{subfig: d}, where $\beta := 1/w_1^-(0.5) = 1.414$.
	Note that the payoffs for player $1$ are negations of her payoffs in (b), and her probability weighing function for losses is same as her probability weighing function for gains in (b).
	Thus her CPT value function $V_1(L_1(p,q))$ is the negation of her CPT value function in (b).
	In particular, we have $V_1(L_1(0,0.5)) = V_1(L_1(1,0.5)) > V_1(L_1(p,0.5))$ for all $p \in (0,1)$.
	Thus, $0.5 \in \co(\BBR_1(\mu_{-1}(0.5)))$, but $0.5 \notin \BBR_1(\mu_{-1}(0.5))$.
	The payoffs and CPT features of player $2$ are same as in (b).
	Thus, $(0.5,0.5) \in (\mNE \cap \mBBNE) \back \BBNE$.

	\item Let the CPT features for both the players be as in (b).
	Let the payoff matrix be as shown in figure~\ref{subfig: e}, where $\beta := 1/w_1^+(0.5) = 1.414, \epsilon = 0.1$, and $\gamma := (1-\tilde p)/\tilde p$; here $\tilde p = 0.582$ is the unique maximizer of $V_1(L_1(p, 0.5))$ (see figure~\ref{plot: region_e}).
	We have $V_1(L_1(0,0.5)) = 2.071 > 2 = V_1(L_1(1, 0.5))$ and $\argmax_p V_1(L_1(p, 0.5)) = \{\tilde p\}$ with $V_1(L_1(\tilde p, 0.5)) = 2.125$.
	From the definition of $\gamma$, we see that player $2$ is indifferent between her two actions, given her belief $\tilde p$ over player $1$'s actions.
	Thus, $(\tilde p, 0.5) \in \BBNE \back \mNE$.

	\item Let the CPT features be as in example~\ref{ex: No blackboxNE}.
	Let $p^* = 0.996$ and $q^* = 0.340$ be the same as in example~\ref{ex: No blackboxNE}.
	Let the payoff matrix be as shown in figure~\ref{subfig: f}.
	Note that the payoffs for both the players are the same as in example~\ref{ex: No blackboxNE}.
	Recall 
	$\BBR_1(\mu_{-1}(q)) = 0$ for $q < q^*$, 
$\BBR_1(\mu_{-1}(q)) = \{0, p^*\}$ for $q = q^*$, and 
$\BBR_1(\mu_{-1}(q)) \subset [p^*, 1]$ for $q > q^*$,
 and hence $0.5 \in \co(\BBR_1(\mu_{-1}(q^*)))$ and $0.5 \notin \BBR_1(\mu_{-1}(q^*))$.
Further, from the definition of $\gamma$, we have $V_2(L_2(0.5,0)) = V_2(L_2(0.5,q)) = V_2(L_2(0.5, 1)), \forall q \in [0,1]$.
Hence, $(0.5, q^*) \in \mBBNE \back (\mNE \cap \BBNE)$.

\item Finally, if we let the players have EUT preferences and the payoffs as in (a), then $(1,0) \notin (\mNE \cup \mBBNE)$.
\end{enumerate}
\begin{figure}
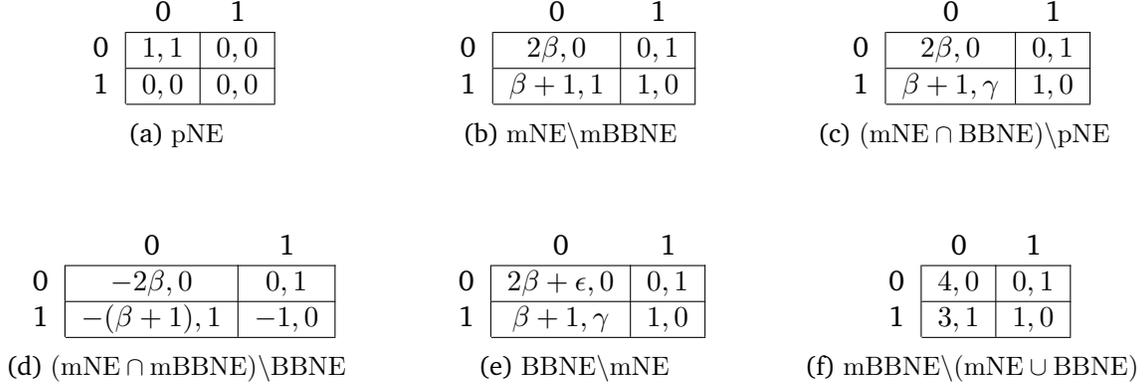

\parbox{0.3\linewidth}{
\centering
\begin{tabular}{c | c | c |}
	 \multicolumn{1}{c}{}	& \multicolumn{1}{c}{0}   &  \multicolumn{1}{c}{1}\\
	 \cline{2-3}
	 0 & $1,1$  &  $0,0$\\
	 \cline{2-3}
	 1	& $0,0$ & $0,0$ \\
	 \cline{2-3}
	 \end{tabular}
	 \subcaption{$\pNE$}
	 \label{subfig: a}
	 }
\hfill
\parbox{0.3\linewidth}{
\centering
\begin{tabular}{c | c | c |}
	 \multicolumn{1}{c}{}	& \multicolumn{1}{c}{0}   &  \multicolumn{1}{c}{1}\\
	 \cline{2-3}
	 0 & $2\beta, 0$  &  $0, 1$\\
	 \cline{2-3}
	 1	& $\beta + 1, 1$ & $1, 0$ \\
	 \cline{2-3}
	 \end{tabular}
	 \subcaption{$\mNE \back \mBBNE$}
	 \label{subfig: b}
}
\hfill
\parbox{0.3\linewidth}{
\centering
\begin{tabular}{c | c | c |}
	 \multicolumn{1}{c}{}	& \multicolumn{1}{c}{0}   &  \multicolumn{1}{c}{1}\\
	 \cline{2-3}
	 0 & $2\beta, 0$  &  $0,1$\\
	 \cline{2-3}
	 1	& $\beta + 1, \gamma$ & $1, 0$ \\
	 \cline{2-3}
	 \end{tabular}
	 \subcaption{$(\mNE \cap \BBNE) \back \pNE$}
	 \label{subfig: c}
}
\vspace{1cm}\\
\parbox{0.3\linewidth}{
\centering
\begin{tabular}{c | c | c |}
	 \multicolumn{1}{c}{}	& \multicolumn{1}{c}{0}   &  \multicolumn{1}{c}{1}\\
	 \cline{2-3}
	 0 & $-2\beta, 0$  &  $0,1$\\
	 \cline{2-3}
	 1	& $-(\beta + 1), 1$ & $-1, 0$ \\
	 \cline{2-3}
	 \end{tabular}
	 \subcaption{$(\mNE \cap \mBBNE) \back \BBNE$}
	 \label{subfig: d}
	 }
\hfill
\parbox{0.3\linewidth}{
\centering
\begin{tabular}{c | c | c |}
	 \multicolumn{1}{c}{}	& \multicolumn{1}{c}{0}   &  \multicolumn{1}{c}{1}\\
	 \cline{2-3}
	 0 & $2\beta + \epsilon, 0$  &  $0,1$\\
	 \cline{2-3}
	 1	& $\beta + 1, \gamma$ & $1, 0$ \\
	 \cline{2-3}
	 \end{tabular}
	 \subcaption{$\BBNE \back \mNE$}
	 \label{subfig: e}
}
\hfill
\parbox{0.3\linewidth}{
\centering
\begin{tabular}{c | c | c |}
	 \multicolumn{1}{c}{}	& \multicolumn{1}{c}{0}   &  \multicolumn{1}{c}{1}\\
	 \cline{2-3}
	 0 & $4, 0$  &  $0,1$\\
	 \cline{2-3}
	 1	& $3, 1$ & $1, 0$ \\
	 \cline{2-3}
	 \end{tabular}
	 \subcaption{$\mBBNE \back (\mNE \cup \BBNE )$}
	 \label{subfig: f}
}
	 \caption{Payoff matrices for the
	 $2\times2$ games
	 in example~\ref{ex: non-empty_regions}. 
	 The rows and the columns correspond to the actions of player $1$ and player $2$, respectively. 
	 In each cell, the left and right entries correspond to player $1$ and player $2$, respectively.
	 The labels indicate the corresponding regions in figure~\ref{fig: venn_diagram}. 
	 The game matrix for the example corresponding to region (g) is the same as that for the one corresponding to region (a).
     }\label{fig: non-empty regions}
\end{figure}
 \begin{figure}
 \parbox{0.45\linewidth}{
 \input{plots/Value_plot_example_b}
 }
 \hfill
 \parbox{0.45\linewidth}{
 \input{plots/Value_plot_example_e}
 }
 \end{figure}

 \qed
 \end{example}
%




\section{Conclusion}
\label{sec: conclusion}

In the study of non-cooperative game theory from a decision-theoretic viewpoint, it is important to distinguish between 
two types of randomization: 
\begin{enumerate}
	\item conscious randomizations implemented by the players, and
	\item randomizations in conjectures resulting from the beliefs held by the other players about the behavior of a given player.
\end{enumerate}
This difference becomes evident when the preferences of the players over lotteries do not satisfy betweenness, a weakened form of independence property.
We considered $n$-player normal form games where players have CPT preferences,
an important example of preference relation that does not satisfy betweenness.
This gives rise to four types of Nash equilibrium notions,
depending on the different types of randomizations.
We defined these different notions of 
equilibrium
and 
discussed the question of their existence. 
The results are summarized in table~\ref{tab: summary}.


\begin{table}
\centering
\renewcommand{\arraystretch}{1.2}
\begin{tabular}{>{\raggedright}p{4.5cm} >{\raggedright}p{2.2cm} >{\raggedright}p{3.3cm} >{\raggedright\arraybackslash}p{2.5cm}}
\toprule
Type of Nash equilibrium & Strategies & Conjectures & Always exists \\ 
\toprule
Pure Nash equilibrium & Pure actions & Exact conjectures & No \\ \midrule
Mixed action Nash equilibrium & Pure actions & Mixed conjectures &  Yes \citep{keskin2016equilibrium}\\	\midrule
Black-box strategy Nash equilibrium & Black box strategies & Exact conjectures & No (Example~\ref{ex: No blackboxNE})\\ \midrule
Mixed black-box strategy Nash equilibrium & Black box strategies & Mixed conjectures & Yes (Theorem~\ref{thm: BBNEexists}) \\
\bottomrule
\end{tabular}
\caption{Different types of Nash equilibrium when players have CPT preferences.}
\label{tab: summary}
\end{table}
\appendix

\section{Joint continuity of the concave hull of a jointly continuous function}
\label{app: continuity}

Let $\Delta^{m-1}$ and $\Delta^{n-1}$ be simplices of the corresponding dimensions with the usual topologies. Let $f: \Delta^{m-1} \times \Delta^{n-1} \to \bbR$ be a continuous function on $\Delta^{m-1} \times \Delta^{n-1}$ (with the product topology).
Let $\cal{P}(\Delta^{m-1})$ denote the space of all probability measures on $\Delta^{m-1}$ with the topology of weak convergence. 
Let $g : \Delta^{m-1} \times \Delta^{n-1} \to \bbR$ be given by
\[
	g(x,y) := \sup \l\{\bbE_{X \sim p} f(X,y) \big| p \in \cal{P}(\Delta^{m-1}), \bbE_{X \sim p} \id(X) = x \r\}.
\]
where $\id: \Delta^{m-1} \to \Delta^{m-1}$ is the identity function $\id (x) := x, \forall x \in \Delta^{m-1}$ and the expectation is over a random variable $X$ taking values in $\Delta^{m-1}$ with the distribution $p$.

\begin{proposition}
	The function $g(x,y)$ is continuous on $\Delta^{m-1} \times \Delta^{n-1}$.
\end{proposition}
\begin{proof}
	We first prove that the function $g(x,y)$ is upper semi-continuous.
	Let $x_t \to x$ and $y_t \to y$. Let $\{g(x_{t_n},y_{t_n})\}$ be a convergent subsequence of $\{g(x_t,y_t)\}$ with limit $L$. It is enough to show that the limit $L \leq g(x,y)$. Since
	for all $n$ the set
	$\{p \in \cal{P}(\Delta^{m-1}), \bbE_{X \sim p} \id(X) = x_{t_n}\}$
	is compact,
	we know that there exists $p_{t_n} \in \cal{P}(\Delta^{m-1})$, 
	such that $g(x_{t_n},y_{t_n}) = \bbE_{X \sim p_{t_n}} [f(X,y_{t_n})]$ and $\bbE_{X \sim p_{t_n}} [\id (X)] = x_{t_n}$. 
	The sequence $\{p_{t_n} \}$ has a convergent subsequence, say $p_{t_{n_k}} \to \bar p$ (because $\cal{P}(\Delta^{m-1})$ is a compact space). 
	Now, $\bbE_{X \sim \bar p} [\id (X)] = \lim_k \bbE_{X \sim p_{t_{n_k}}} [\id (X)] = \lim_k x_{t_{n_k}} = x$. 
	Further, $\bbE_{X \sim p_{t_{n_k}}} [f(X,y_{t_{n_k}})] \to \bbE_{X \sim \bar p} [f(X,y)]$, 
	since the product distributions $p_{t_{n_k}} \times \1\{ y_{t_{n_k}}\}$, for all $k$, on $\Delta^{m-1} \times \Delta^{n-1}$,  
	converge weakly to the product distribution $\bar p \times \1 \{y\}$.
	Thus, $L = \bbE_{X \sim \bar p} [f(X,y)] \leq g(x,y)$ and the function $g(x,y)$ is upper-semicontinuous.

	We now prove that the function $g(x,y)$ is lower semi-continuous.
	Let $x_t \to x$ and $y_t \to y$. The simplex $\Delta^{m-1}$ can be triangulated into finitely many other simplices, say $T_1,\dots,T_k$, whose vertices are $x$ and some $m-1$ of the $m$ vertices of $\Delta^{m-1}$.
	Let $(x_{t_n})$ be any subsequence such that all $x_{t_n} \in T_j$ for some simplex. It is enough to show that the $\liminf$ of the sequence $\{g(x_{t_n},y_{t_n})\}$ is greater than or equal to $g(x,y)$. Let the other vertices of $T_j$ be $e_1, \dots , e_{m-1}$. Let $z_{t_n} = (z_{t_n}^1,\dots,z_{t_n}^l)$ be the barycentric coordinates of $x_{t_n}$ with respect to the simplex $T_j$, i.e.
	\[
		x_{t_n} = (1 - z_{t_n}^1 - \dots - z_{t_n}^{m-1})x + z_{t_n}^1 e_1 + \dots + z_{t_n}^{m-1} e_{m-1}.
	\]
  The function $g(x,y)$ is concave in $x$ for any fixed $y$ by construction.
  We have,
	\[
		g(x_{t_n},y_{t_n}) \geq (1 - z_{t_n}^1 - \dots - z_{t_n}^{m-1})g(x,y_{t_n}) + z_{t_n}^1 g(e_1,y_{t_n}) + \dots + z_{t_n}^{m-1} g(e_{m-1},y_{t_n}).
	\]
	Since $z_{t_n} \to (0,\dots,0)$ and $g(e_1,y_{t_n}), \dots,g(e_{m-1},y_{t_n})$ are all finite we get,
	\[
		\liminf g(x_{t_n},y_{t_n}) \geq \liminf g(x,y_{t_n}).
	\]
	Let 
	$\tilde p \in \cal{P}(\Delta^{m-1})$ 
	be such that $\bbE_{X \sim \tilde p} [f(X,y)] = g(x,y)$ and $\bbE_{X \sim \tilde p} [\id (X)] = x$. Then, $g(x,y_{t_n}) \geq \bbE_{X \sim \bar p} [f(X,y_{t_n})]$, for all $n$, and hence,
	\[
		\liminf g(x,y_{t_n}) \geq \liminf \bbE_{X \sim \tilde p} [f(X,y_{t_n})] = g(x,y).
	\]
	This shows that the function $g(x,y)$ is lower semi-continuous.
	
	Since the function $g(x,y)$ is upper and lower semi-continuous, it is continuous.
\end{proof}


\section{An interesting functional equation}
\label{sec: proofs}

\begin{lemma}
\label{lem: w_functionaleq}
	Let $w: [0,1] \to [0,1]$ be a continuous, strictly increasing function such that $w(0) = 0$ and $w(1) = 1$.
	For any $0 \leq a_1 < c_1 < b < c_2 < a_2 \leq 1$ such that $(a_2 - b)(b - c_1) = (b - a_1)(c_2 - b)$, let
 \begin{align}
 \label{eq: ratio_function}
	\l[w(a_2) - w(b)\r]&\l[w(b) - w(c_1)\r] =  \l[w(b) - w(a_1)\r]\l[w(c_2) - w(b)\r].
\end{align}	
Then $w(p) = p$ for all $p \in [0,1]$.
\end{lemma}
\begin{proof}
	Taking 
	$a_1 = 0, c_1 = 1/4, b = 1/2, c_2 = 3/4$ and $a_2 = 1$
	in \eqref{eq: ratio_function}
	we get,
 \begin{align*}
	\l[1 - w(1/2)\r]&\l[w(1/2) - w(1/4)\r] =  \l[w(1/2)\r]\l[w(3/4) - w(1/2)\r],
\end{align*}
and hence,
\[
	w(3/4) = \frac{w(1/2) + w(1/2)w(1/4) - w(1/4)}{w(1/2)}.
\]
Note that $w(1/2) > 0$.
Taking 
	$a_1 = 0, c_1 = 1/4, b = 1/3, c_2 = 1/2$ and $a_2 = 1$
in \eqref{eq: ratio_function} 
we get,
 \begin{align*}
	\l[1 - w(1/3)\r]&\l[w(1/3) - w(1/4)\r] = \l[w(1/3)\r]\l[w(1/2) - w(1/3)\r],
\end{align*}
and hence,
	\[
		w(1/3) = \frac{w(1/4)}{1 - w(1/2) + w(1/4)}.
	\]
Note that $1 - w(1/2) + w(1/4) > 1 - w(1/2) > 0$.
Taking 
	$a_1 = 0, c_1 = 1/3, b = 1/2, c_2 = 2/3$ and $a_2 = 1$
in \eqref{eq: ratio_function}
we get,
 \begin{align*}
	\l[1 - w(1/2)\r]&\l[w(1/2) - w(1/3)\r] =  \l[w(1/2)\r]\l[w(2/3) - w(1/2)\r],
\end{align*}
and substituting for $w(1/3)$ we get,
	\[
		w(2/3) = \frac{w(1/2) - w(1/2)^2 + 2w(1/2)w(1/4) - w(1/4)}{w(1/2) - w(1/2)^2 + w(1/2)w(1/4)}.
	\]
Note that 
$$w(1/2) - w(1/2)^2 + w(1/2)w(1/4) = w(1/2)[1 - w(1/2) + w(1/4)] > 0.$$
Taking 
	$a_1 = 0, c_1 = 1/2, b = 2/3, c_2 = 3/4$ and $a_2 = 1$
in \eqref{eq: ratio_function} 
we get,
 \begin{align*}
	\l[1 - w(2/3)\r]&\l[w(2/3) - w(1/2)\r] =  \l[w(2/3)\r]\l[w(3/4) - w(2/3)\r].
\end{align*}
Simplifying we get,
\begin{align*}
	w(2/3) - w(2/3)w(3/4) = w(1/2) - w(1/2)w(2/3), 
\end{align*}
Substituting for $w(2/3)$ and $w(3/4)$ we get,
\begin{align*}
	&\l[\frac{w(1/2) - w(1/2)^2 + 2w(1/2)w(1/4) - w(1/4)}{w(1/2) - w(1/2)^2 + w(1/2)w(1/4)}\r] \l[\frac{w(1/4) - w(1/2)w(1/4)}{w(1/2)}\r] \\
	&= w(1/2) \l[\frac{w(1/4) - w(1/2)w(1/4)}{w(1/2) - w(1/2)^2 + w(1/2)w(1/4)} \r].
\end{align*}
Since $w(1/4) - w(1/2)w(1/4) > 0$ and $w(1/2) - w(1/2)^2 + w(1/2)w(1/4) > 0$, we get
\begin{align*}
	w(1/2) - w(1/4) = 2w(1/2)[w(1/2) - w(1/4)].
\end{align*}
 Since $w(1/2) - w(1/4) > 0$, we get
$w(1/2) = 1/2$.

For any fixed $0 \leq x < y \leq 1$,
let 
	$$w'(p') := \frac{w(p'(y-x) + x) - w(x)}{w(y) - w(x)}, \text{ for all } 0 \leq p' \leq 1.$$ 
	Note that $w':[0,1] \to [0,1]$ is a continuous, strictly increasing function with $w'(0) = 0$ and $w'(1) = 1$.
	Further, if $0 \leq a'_1 < c'_1 < b' < c'_2 < a'_2 \leq 1$ are such that $(a'_2 - b')(b' - c'_1) = (b' - a'_1)(c'_2 - b')$, then
	 \begin{align*}
	\l[w'(a'_2) - w'(b')\r]&\l[w'(b') - w'(c'_1)\r] = \l[w'(b') - w'(a'_1)\r]\l[w'(c'_2) - w'(b')\r].
\end{align*}	
Thus $w'(1/2) = 1/2$ and hence
$w\l((x + y)/{2}\r) = (w(x) + w(y))/{2}.$
Using this repeatedly we get
$w({k}/{2^t}) = {k}/{2^t},$
for $0 \leq k \leq 2^t$, $t = 1, 2, \dots$.
Continuity of $w$ then implies $w(p) = p$, for all $p \in [0,1]$.
\end{proof}

\bibliographystyle{abbrvnat} 
\bibliography{Bib_Database}

\end{document}